\documentclass{article}

\usepackage{amsmath,amsthm}

\newtheorem{thm}{Theorem}

\theoremstyle{definition}

\theoremstyle{remark}

\newcommand{\thmref}[1]{Theorem~\ref{#1}}

\newcommand{\cpr}[2]{\mathbf{c}_{#1}^{#2}}
\newcommand{\cpro}{\cpr{}{\textsc{o}}}
\newcommand{\cprr}{\cpr{}{r}}

\newcommand{\tmi}[2]{\mathbf{M}_{#1}^{#2}}
\newcommand{\tmio}{\tmi{}{o}}

\newcommand{\Pp}{\mathcal{P}}
\newcommand{\NP}{\mathcal{NP}}
\newcommand{\pnpvs}{$\Pp vs.\NP$}
\newcommand{\pnpeq}{$\Pp=\NP$}
\newcommand{\pnpnq}{$\Pp\neq\NP$}

\newcommand{\Pon}{P_{1}}
\newcommand{\Ptw}{P_{2}}
\newcommand{\Pth}{P_{3}}
\newcommand{\Qon}{Q_{1}}
\newcommand{\Qtw}{Q_{2}}
\newcommand{\Qth}{Q_{3}}

\newcommand{\dsat}{$D_{sat}$}
\newcommand{\ndsat}{$ND_{sat}$}
\newcommand{\trt}{$\textbf{t}$}

\newcommand{\dtm} [1] {DTM{#1}}
\newcommand{\ndtm} [1] {NDTM{#1}}
\newcommand{\apmachine} [1] {APM{#1}}

\newcommand{\Pproblems} [1] {P problems{#1}}
\newcommand{\npcomp} [1] {NPC{#1}}

\begin{document}
\title{\pnpnq $ $ by Modus Tollens}
\author{Joonmo Kim \\ Applied Computer Engineering Department, Dankook University \\ South Korea \\ \texttt{chixevolna@gmail.com}}
\maketitle

\bibliographystyle{plain}

\begin{abstract}
An artificially designed Turing Machine algorithm $\tmio$
generates the instances of the satisfiability problem, and
check their satisfiability.
Under the assumption \pnpeq,
we show that
$\tmio$ has a certain property,
which,
without the assumption,
$\tmio$ does not have.
This leads to \pnpnq $ $ by modus tollens.
\end{abstract}


\paragraph{Notice}
Version 2 of this article, \cite{j}, is the final writing for the proof. Later versions are for adding meta data and replies. Version 5 is for adding comments for easier readings: subsubsection 0.1.1 and subsection 0.2. Version 6 is for modifying subsection 0.2. Version 7 is for adding Addendum of \cite{j} at page 8.

\section{Introduction}
This formula $(\Pon \rightarrow (\Ptw \rightarrow \Pth)) \wedge (\neg (\Ptw \rightarrow \Pth))$
concludes $\neg \Pon$ by modus tollens.
As an example, we may easily prove that: it is false that each of all numbers can be described by a
finite-length digits. Simple proof is somehow showing the number of infinite-length digits.
Instead, a boy makes use of the modus tollens as follows.
He lets the propositions be:
\\$\Pon$: each of all numbers can be described by a finite-length digits,
\\$\Ptw$: there exists the irrational number,
\\$\Pth$: the irrational number can be described by a finite-length digits.

Then we know that $\Pon \rightarrow (\Ptw \rightarrow \Pth)$ can be proved,
and so can $\neg (\Ptw \rightarrow \Pth)$.

In proving \pnpnq, we are going to follow his proof.
Let $\Pon$ be \pnpeq,
$\Ptw$ an argument that an algorithm (say $\tmio$) exists, and
$\Pth$ an argument on the property of $\tmio$.
It will be seen that if \pnpeq $ $ then $\tmio$ has the property of $\Pth$,
otherwise it does not.

\section{Algorithm $\tmi{}{}$ and Cook's Theory}
Cook's Theory
\cite{gajo:ci}
says that
the accepting computation of
a non-deterministic poly-time Turing Machine on an input $x$
can be transformed in a polynomial time to a satisfiable instance,
denoted as $\cpr{}{}$,
of the satisfiability problem(SAT).

All the conditions of being an accepting computation are expressed, in $\cpr{}{}$, as a collection of Boolean clauses.
There are six groups of clauses in $\cpr{}{}$. Quoted from \cite{gajo:ci}, each group imposes restrictions as: \\
$G_{1}$: at each time $i$, Turing Machine $M$ is in exactly one state, \\
$G_{2}$: at each time $i$, the read-write head is scanning exactly one tape square, \\
$G_{3}$: at each time $i$, each tape square contains exactly one symbol from $\Gamma$, \\
$G_{4}$: at time $0$, the computation is in the initial configuration of its checking stage for input $x$, \\
$G_{5}$: by time $p(n)$, $M$ has entered state $q_{y}$ and hence has accepted $x$, \\
$G_{6}$: for each time $i$, $0 \leq i < p(n)$, the configuration of $M$ at time $i+1$ follows by a single application of the transition function $\delta$ from the configuration at time $i$.

Though the groups of clauses are designed to represent the runs of
poly-time Turing Machines,
they can be modified to represent the runs of longer-time ones.
As long as there exists any finite-length accepting computation path from a problem instance to an accepting state over a Turing Machine,
how long the path may be,
there may exist the corresponding clauses for each of all the transitions along the path.

Therefore, we may extend the meaning of an accepting computation as the representation of a finite run of a Turing Machine on one of its accepting input, regardless of the run time.
Note that, in \cite{gajo:ci}, the time for the transformation from an accepting computation to the corresponding SAT instance should have been kept within a polynomial, but the proof here does not regard any kind of run time, as long as the run is finite.

We may observe that
the clauses in
$\cpr{}{}$ can be divided into two parts: one is for the representation of
the given input, and the other the run of the Turing Machine on the input.
Let the \textit{input-part} be the clauses for the description of the input $x$, which is denoted as $\cpr{}{x}$.
Let the \textit{run-part} (denoted as $\cprr$ ) be the part of $\cpr{}{}$
such that the clauses for $\cpr{}{x}$
are cut off from $\cpr{}{}$, i.e.,
$\cprr$ is a part of $\cpr{}{}$ that actually represents the operations of the corresponding Turing Machine (grouped in \cite{gajo:ci} as $G_{1}$, $G_{2}$, $G_{3}$, $G_{5}$, $G_{6}$).

Let $\tmi{}{}$ be a Turing Machine algorithm, to which the input is a string denoted by $y$.
$\tmi{}{}$ is designed to include a finite number of $\cprr$'s (i.e., $\cpr{1}{r}, \cdots, \cpr{n}{r}$),
which are
trimmed from $\cpr{}{}$'s (i.e., $\cpr{1}{}, \cdots, \cpr{n}{}$),
where
$\cpr{}{}$'s are
arbitrarily selected accepting computations.
That is, $\cpr{i}{r}$  is formed by cutting off
$\cpr{i}{x}$ from $\cpr{i}{}$,
where $x$ is the input (represented by the initial configuration in $G_{4}$) of the computation that corresponds to $\cpr{i}{}$.
According to the run-parts included, countably many $\tmi{}{}$'s can be constructed, and they can be somehow ordered as:
$\tmi{1}{}$, $\tmi{2}{}$, $\cdots$, $\tmi{i}{}$, $\cdots$.
For $\tmi{i}{}$, its run-parts can be written as:
$\cpr{i1}{r}, \cpr{i2}{r}, \cdots, \cpr{im}{r}$.

Given an input $y$,
during the run of $\tmi{i}{}$,
$\cpr{}{y}$ and $\cpr{ij}{r}$ are concatenated,
forming $\cpr{ij}{}$ ($1 \leq j \leq m$).
For each $\cpr{ij}{}$, the module of SAT-solver,
which will accordingly be chosen to be
either a deterministic algorithm or non-deterministic,
performs the run for the satisfiability check.
If $\cpr{ij}{}$ is satisfiable
then $\tmi{i}{}$
increases its counter, and goes on to the next $\cpr{}{}$;
this process repeats up to $\cpr{im}{}$.
At the end of the counting, $\tmi{i}{}$ accept $y$ if the counter holds an odd number.
That is, the task of $\tmi{}{}$, at the given input of a finite string $y$, is
to count the number of satisfiable
$\cpr{}{}$'s, and accept $y$ if the counted number is odd.
Observe that it is not impossible for
$\cpr{ij}{}$ to be satisfiable
though the chances are usually rare,
and that the run time of $\tmi{i}{}$ shall be finite
because each of all the satisfiability check will take a finite time.
These observations ensure that
$\tmi{i}{}$ determines a set of acceptable strings,
supporting $\tmi{i}{}$ to be a Turing Machine.

\section{Algorithm $\tmio$ and the Theorem}
Let a \textit{particular transition table} of a Turing Machine
be a transition table particularly for
just one or two accepting problem instances,
where the computation time does not matter as long as it is finite.
So, a particular transition table
for an accepting problem instance
may produce an accepting computation
by running on a Turing Machine.
Observe that each of all accepting computations may have its particular transition table, i.e.,
the table can be built
by collecting all the distinguished transitions from the computation,
where we know that
a computation is a sequence of the transitions of configurations of a Turing Machine.

We then introduce an algorithm, denoted as $\tmio$,
which is one of $\tmi{}{}$'s with the property as follows.
Let $\widehat{ac}_{\tmi{}{}}$ be
the accepting computation of the run of $\tmi{}{}$ on an input $y$,
let \trt $ $ be a particular transition table for $\widehat{ac}_{\tmi{}{}}$,
let $\cpro$ be one of $\cpr{}{}$'s that appear during the run of $\tmi{}{}$, and
let $\widehat{ac}_{\cpro}$ be the accepting computation, which is described by the clauses of $\cpro$.
If \trt $ $ is also a particular transition table for $\widehat{ac}_{\cpro}$
then denote $\tmi{}{}$ as $\tmio$.

For the run of $\tmio$, we may consider two types of particular transition tables.
Let \dsat $ $ be the particular transition table, by which
$\tmio$ runs deterministically and
the SAT-solver module runs deterministically in a poly-time for the length of $\cpr{}{}$.
Analogously, we may have \ndsat, by which
$\tmio$ runs non-deterministically and
the SAT-solver module runs non-deterministically in a poly-time for the length of $\cpr{}{}$.

\begin{thm}\label{thm-main}
\pnpnq
\end{thm}
\begin{proof}
Let $\Pon$, $\Ptw$ and $\Pth$ be the following propositions: \\
$\Pon$: \pnpeq, \\
$\Ptw$: $\tmio$ exists, \\
$\Pth$: there exists \trt, which is \dsat.

By modus tollens, $(\Pon \rightarrow (\Ptw \rightarrow \Pth)) \wedge (\neg (\Ptw \rightarrow \Pth))$ may conclude $\neg \Pon$.

1) Proposition $\Pon \rightarrow (\Ptw \rightarrow \Pth)$ is to show that
if $\tmio$ exists then
there exists \trt, which is \dsat,
when \pnpeq $ $ is the antecedent.

By \pnpeq, there exists a deterministic poly-time SAT-solver, so the SAT-solver module in $\tmio$
can be implemented to be
a deterministic algorithm, which runs in a poly-time for the length of $\cpr{}{}$.
Thus, \trt $ $ can be \dsat.

2) For the latter part, we are to show $\neg (\Ptw \rightarrow \Pth)$.
By the inference rule,
if $\Ptw$ is true then
a contradiction from $\Ptw \rightarrow \Pth$
implies $\neg (\Ptw \rightarrow \Pth)$.

We can show that $\Ptw$ is true, as follows.
For any chosen $\cpro$,
build two non-deterministic particular transition tables for
$\widehat{ac}_{\tmio}$ and $\widehat{ac}_{\cpro}$ separately, and then merge the two
so that
one of the two computations can be
chosen selectively
from the starting state during the run.
$\tmio$ may exist by this \trt, which is \ndsat.

Next we are to derive a contradiction from $\Ptw \rightarrow \Pth$.
Observe that
$\Ptw \rightarrow \Pth$ implies specifically this $argument$:
if $\tmio$ exists then
there exists \trt, which is \dsat $ $ particular transition table
for both $\widehat{ac}_{\tmio}$ and $\widehat{ac}_{\cpro}$.
Then, by the $argument$
together with the
definition (of $\tmio$) that
$\widehat{ac}_{\tmio}$ and $\widehat{ac}_{\cpro}$ are the accepting computations that share the same input,
it is concluded that both $\widehat{ac}_{\tmio}$ and $\widehat{ac}_{\cpro}$ are exactly the same computation, i.e., all the transitions of
the configurations of $\widehat{ac}_{\tmio}$ and
those of $\widehat{ac}_{\cpro}$ are exactly the same.

According to \cite{gajo:ci},
with the assignment of the truth-values, which are acquired by the SAT-solver module,
to the the clauses of $\cpro$,
we later may acquire
$\widehat{ac}_{\cpro}$ in the form of
the sequence of transitions of configuration of a Turing Machine.

Now, let $i$ be the number of
the transitions between the configurations in
$\widehat{ac}_{\tmio}$,
$j$ the number of the clauses of $\cpro$, and
$k$ the number of
the transitions between the configurations in $\widehat{ac}_{\cpro}$.

Since,
at the least, all the clauses of $\cpro$ should once be loaded on the tape of the Turing Machine as well as other $\cpr{}{}$'s,
we may have $i > j$, and
since it is addressed in \cite{gajo:ci} that
each transition of an accepting computation is described by more than one clauses, we may have $j > k$,
resulting $i > j > k$.

However, $i = k$ because it is concluded above that
both $\widehat{ac}_{\tmio}$ and $\widehat{ac}_{\cpro}$ are exactly the same computation.
This is a contradiction from $(\Ptw \rightarrow \Pth)$.
\end{proof}

As a commentary, to make sure that the proof does not fall into the similar oddity of the well-known relativizations of \pnpvs,
it'd be better to consider this argument: \\
$(\Qon \rightarrow (\Qtw \rightarrow \Qth)) \wedge (\neg (\Qtw \rightarrow \Qth))$, where \\
$\Qon$: \pnpnq, \\
$\Qtw$: $\tmio$ exists, \\
$\Qth$: there exists \trt, which is \ndsat.

Similarly, by the antecedent \pnpnq,
there exist only non-deterministic algorithm for the SAT-solver module, so
$\Qon \rightarrow (\Qtw \rightarrow \Qth)$
can analogously be proved.
For the latter part,
$\neg (\Qtw \rightarrow \Qth)$ become true
if $\Qtw \rightarrow \Qth$ implies $i = k$, as in the proof.
If so, we have the result that \pnpeq.

However, $\Qtw \rightarrow \Qth$ may imply $i = k$
if $\widehat{ac}_{\tmio}$ and $\widehat{ac}_{\cpro}$
are the same,
but we know, referring to the proof, that
there is no such \trt $ $ that makes
$\widehat{ac}_{\tmio}$ and $\widehat{ac}_{\cpro}$
the same.
In fact, we may build many \trt's that does not incur $i = k$
from $\Qtw \rightarrow \Qth$,
as mentioned in the proof.

\clearpage
\setcounter{section}{0}
\section*{Replies to the Critiques}
Followings are the replies to the critiques and comments for the proof. Replies to future critiques will be added below.
Author does not expect the success of the proof, rather he is waiting to see what is wrong in \cite{j}.

\subsection{Replies to the Critique of J. Kim's ``P is not equal to NP by Modus Tollens''}
\label{subsection 0.1}

It is quite likely that the Critique of \cite{hassin}
has errors. Authors of \cite{hassin} has a good understanding on \cite{j}, but missed some points.\\

The tables in 2.4 of \cite{hassin} are exactly what the author of \cite{j} expected. (Thanks.)\\

``\textit{3.1 Invalidity of logical argument}'' of \cite{hassin}
said that `\trt $ $ is \dsat' is a fact,
but \trt $ $ can also be \ndsat $ $ especially when the SAT-solver module is implemented non-deterministically.
Notice that the TM algorithms designed for the proof are not practical programming.
As a result, `\trt $ $ is \dsat' is not always true.
In addition, $\tmio$ may exist when \trt $ $ is \ndsat.\\

For \textit{3.2} of \cite{hassin},
the author of \cite{j}
would rather choose ``\textit{3.2.2 Second interpretation}.''
The critique mentioned :
``Note that these accepting computations are not necessarily the same as the accepting computations produced by their respective Turing machines' transition tables.''

Author of \cite{j} agree with this, but, for the proof of \cite{j}, it is enough
if there exist more than one \textbf{case}s that
\textit{the accepting computations are the same as the accepting computations produced by their respective Turing machines’ transition tables}.
The proof derives the contradiction from one of the \textbf{case}s.
The \textbf{case} can be seen by the merged table in 2.4.\\

In 3.3, the critique said ``At the end of his paper, Kim verifies that \ldots.''
The part,
`\textit{the end of his paper},'
is to show that the proof will not be fallen into the relativizations of \pnpvs.
The author of \cite{j}
wanted to show that if the proof is correct then there is no room
for the proof
to be fallen into the relativization matters.
Perhaps, this part can be omitted for now,
while questioning on the correctness of the proof itself.

\paragraph{Comments}
Authors of \cite{hassin} show a good understanding on \cite{j}, but have the errors as above that seem to imply that \cite{j} has not yet been refuted.

\subsubsection{Complementary to Subsection \ref{subsection 0.1}}

Considering the view of the the critique of \cite{hassin},
author of \cite{j} would like to add more explanation to guide other readers for easier and quicker  understanding.

First of all, author of \cite{j} does not think
that the proof of \thmref{thm-main} in \cite{j}
has to be
divided into two interpretations:
\textit{3.2.1 First interpretation} and
\textit{3.2.2 Second interpretation} as in \cite{hassin}.

It seems that
authors of \cite{hassin} had the two interpretations
because
they think that
\trt $ $ and the two accepting computations,
$\widehat{ac}_{\tmi{}{}}$ and $\widehat{ac}_{\cpro}$,
should be derived from their respective Turing machines' transition tables
so that the property $i > j > k$ can be applied correctly, to ensure a consistent and coherent frame for the proof, in the proof of \thmref{thm-main}.

However, such a frame does not let
the proof of \thmref{thm-main} reach the completion: this is what
authors of \cite{hassin} claim.

If the proof of \thmref{thm-main} is correct,
author of \cite{j} is claiming a larger frame, in which
a consistent and coherent proof can be ensured.

We see in the proof that there are two ways that \trt's are made.
Firstly, in part 1) of the proof, \trt $ $, which is \dsat, comes into existence
since there exist
the respective Turing machine and its transition table. In fact, this \trt $ $ is a special case of the transition table of the respective Turing machine, $\tmio$.
As mentioned, $\tmi{}{}$, which is the base of $\tmio$, is composed of
the SAT-solver module and the remaining part,
where the module and the part
can be implemented either way deterministically or non-deterministically.
Actually, since the remaining part of $\tmi{}{}$ can be programmed simply, it can easily be implemented deterministically. Thus, by the antecedent
$\Pon$, all parts of $\tmio$ can be implemented deterministically,
resulting the existence of \trt $ $, which is \dsat.
The second of making \trt $ $ is the way that the tables are merged, which is rephrased in
 2.4 of \cite{hassin}.
So, this \trt $ $ is not related to the respective Turing machine.

By the way, the property $i > j > k$ is the one for the relationship between the accepting computations, $\widehat{ac}_{\tmio}$ and $\widehat{ac}_{\cpro}$, in $\tmio$: nothing about \trt $ $ is considered in $i > j > k$.
Therefore, how \trt $ $ is made does not matter in the proof,
just the existence/non-existence of \trt $ $ for $\tmio$ is important.
This dispenses with the considerations of the two interpretations of \cite{hassin}.

\subsection{Extension of the Proof: Why failed so far}
\label{subsection 0.2}

If the proof of \thmref{thm-main} in \cite{j} is correct, probably we may see why the lower-bound, which is expected to be a deterministic exponential/super-polynomial time, of NP-Complete problems has not yet been found.

By substituting $\Pon$ as ``there exist deterministic exponential/super-polynomial time Turing Machines that solves SAT''
and modifying the definition of \dsat $ $ to be
``the particular transition table, by which
$\tmio$ runs deterministically and
the SAT-solver module runs deterministically in an \textit{exponential/super-polynomial} time for the length of $\cpr{}{}$,''
\thmref{thm-main} can be analogously extended to claim that \textit{there does not exist
any deterministic exponential/super-polynomial time Turing Machine that solves SAT}.

Then, what about the fact that the brute-force search solves SAT?
We may conclude that
the brute-force search is not equal to the deterministic exponential/super-polynomial time Turing Machine algorithm, as follows.

Though the brute-force search for SAT problem can be implemented by a deterministic transition table that causes an exponential time, the run by the table on an input is a non-deterministic computing because the run is
a sequence of generating(guessing) candidate solutions and checking them. Since the corresponding computing is non-deterministic, the particular transition table for the brute-force search is inherently non-deterministic.

So, letting $\Pon$ be ``there exist a brute-force search that solves SAT,''
$\Pon \rightarrow (\Ptw \rightarrow \Pth)$ in the proof is not true,
while it is true when $\Pon$ is as above in this section.

\clearpage
\section*{\begin{center}  Addendum for version2 of this article  \end{center}}
This is a clarification on the unclear part of the proof.
What was unclear is that:
the existence of $\tmio$ can be determined by the existence of the \textit{particular transition table}, but the type of the table was unclear; it can be either the transition table of Deterministic Turing Machine(\dtm{}) or that of Non-Deterministic Turing Machine(\ndtm{}).\\

Consider when \pnpnq{}, only \ndtm{} may embody all the algorithms
of all the Poly-time problems(\Pproblems) because \dtm{} can not embody algorithms of NP-Complete(\npcomp{}) problems.
Let \textit{All Poly Machine}(\apmachine{}) be the Machine that can embody all the algorithms of \Pproblems{}. Actually, either one or both models of Turing Machines(\dtm{}, \ndtm{}) can be designated as \apmachine{}. \apmachine{} varies as follows by the given relationship between $\Pp$ and $\NP$. \\

As mentioned, when \pnpnq{}, (\ndtm{} is \apmachine{}).
When \pnpeq{}, both \dtm{} and \ndtm{} can embody all the algorithms of \Pproblems{},
i.e., ((\ndtm{} is \apmachine{}) $\land$ (\dtm{} is \apmachine{})).
Additionally, when (\pnpnq{}) $\lor$ (\pnpeq{}), (\ndtm{} is \apmachine{}) because \dtm{} can not embody all the algorithms in case (\pnpnq{}), and this can also be seen by the following logical statements:\\
\begin{equation} \nonumber
\begin{split}
\text{(\ndtm{} is \apmachine{})} \lor
\text{((\ndtm{} is \apmachine{})} \land
\text{(\dtm{} is \apmachine{}))} \\
\equiv
\text{(\ndtm{} is \apmachine{})} \land
\text{((\ndtm{} is \apmachine{})} \lor
\text{(\dtm{} is \apmachine{})),}
\end{split}
\end{equation}
we see herein (\dtm{} is \apmachine{}) is negligible.\\

Now, declare that the additional condition of being \textit{particular transition table}(\trt{}) of $\tmio$ is
that: \trt{} can be implemented by the designated \apmachine{}'s.\\

For the proof of $\Pon \rightarrow (\Ptw \rightarrow \Pth)$,
since it is equivalent to $\neg (\Pon \wedge \Ptw) \vee \Pth$, we are to prove $\Pon \wedge \Ptw$ is false.
Then, $\Pon \wedge \Ptw$ is equivalent to $\neg (\Pon \rightarrow \neg \Ptw )$, we are to prove $\Pon \rightarrow \neg \Ptw$ is true,
where $\Pon \rightarrow \neg \Ptw$ is an argument that:
if \pnpeq{} then $\tmio$ does not exist.
As mentioned, when \pnpeq{}, both \dtm{} and \ndtm{} are designated \apmachine{}'s,
and $\tmio$ exists for \ndtm{} as is shown in the proof, but not for \dtm{} as is in the proof too. Therefore, it is true that if \pnpeq{} then $\tmio$ does not exist. \\

For the proof of $\neg(\Ptw \rightarrow \Pth)$, we proved in the proof that $\Ptw$ is true,
where, for better understanding, now we may rewrite $\Ptw$ as:
$1 \rightarrow \Ptw$,  which is equivalent to 
((\pnpnq{}) $\lor$ (\pnpeq{})) $\rightarrow \Ptw$.
Since \ndtm{} is the designated \apmachine{} for ((\pnpnq{}) $\lor$ (\pnpeq{})), we need to prove that the particular transition tables of $\tmio$ can be implemented by \ndtm{},
as is shown in the proof.\\

\end{document}